\documentclass[11pt]{article}
\usepackage[a4paper,margin=1in]{geometry}
\usepackage{amsmath,amssymb,amsthm}
\usepackage{booktabs}
\usepackage{enumitem}
\usepackage[hypertexnames=false]{hyperref}
\usepackage{xspace}
\usepackage{algorithm}
\usepackage{algpseudocode}
\usepackage{tabularx}
\usepackage[T1]{fontenc}
\newcommand{\ZKACE}{ZK-ACE\xspace}
\newcommand{\ID}{\mathsf{ID}}
\newcommand{\msg}{\mathsf{m}}
\newcommand{\rp}{\mathsf{rp}}
\newcommand{\dom}{\mathsf{d}}
\newcommand{\tgt}{\mathsf{t}}
\newcommand{\opnonce}{\mathsf{n}}
\newcommand{\REV}{\mathsf{REV}}
\newcommand{\ctx}{\mathsf{ctx}}
\newcommand{\Adv}{\mathsf{Adv}}
\newcommand{\PPT}{\ensuremath{\mathsf{PPT}}}
\newcommand{\pp}{\mathsf{pp}}
\newcommand{\vk}{\mathsf{vk}}

\newtheorem{definition}{Definition}[section]
\newtheorem{theorem}{Theorem}[section]

\title{\textbf{ZK-ACE: Authorization\\from Registered Identity Commitments}}
\author{
  \texttt{Jian Sheng Wang} \\
  \texttt{jason@acechain.io} \\
  Montreal, Canada
}
\date{November 24, 2026}

\begin{document}
\maketitle

\begin{abstract}
We specify an authorization protocol for high-throughput ledgers whose native
account is a hash of a high-entropy secret rather than a signature public key.
A prover shows knowledge of an opening of that hash and binds the proof to a
registered target, a transaction digest, a public operation nonce, and a
proof-derived replay value. Authorization is one atomic state transition: the
executor checks an authenticated enrollment row, consumes the exact next nonce,
records the replay value, and executes the action. Because the account
identifier is a hash of the root and not a verification key, the identifier
does not change when the installed signature algorithm moves from Ed25519 to
ML-DSA.

Unforgeability of account birth and of later executions reduces to knowledge
soundness, collision resistance of the identity hash, and hardness of
recovering the enrolled secret from the public state and an honestly generated
authorization history; the last assumption relies on the proofs hiding the
witness. Substitution reduces to public-input binding and collision resistance
of the action hash. Replay of a consumed token is a state-machine
impossibility under serializable execution. The experiment does not give the
adversary an adaptive chosen-action proving oracle, so the result is narrower
than EUF-CMA security.

The circuit is a published Poseidon2 sponge over a short witness, not
in-circuit verification of ECDSA or ML-DSA. In the measured reference
snapshot, a Circle STARK serializes to 347,094 bytes (JSON envelope) with
148 bytes of public input; single-threaded release proving and verification on
an Apple M3 Pro take about 15.0\,ms and 2.46\,ms, respectively. On a ledger,
\ZKACE is for account birth and optional off-critical-path aggregation.
Ordinary transfers keep an installed authentication key on the hot path. The
relation does not depend on the argument system; a pairing SNARK can prove the
same statement for existing pairing verifiers.
\end{abstract}


\paragraph{Version note.}
This version matches the construction, security statements, and measurements of the IACR manuscript of November 2026, with additional motivation, related work, and discussion.

\section{Introduction}
\label{sec:intro}

\subsection{Post-quantum signatures and high-throughput authorization}
\label{sec:intro:l1}
NIST has standardized post-quantum signature schemes: the lattice-based
ML-DSA~\cite{mldsa}, the hash-based SLH-DSA~\cite{slhdsa}, and, in draft, the
lattice-based FN-DSA~\cite{falcon}. Their artifacts are much larger than
classical ones (Table~\ref{tab:pqc-sizes}). Every consensus participant must
receive and verify these objects, and the ledger stores them. In rollup-based
systems the transaction data is posted to a base layer whose bandwidth is
explicitly constrained and priced~\cite{eip4844}, so the cost comes from the
size of the authorization artifacts themselves.

\begin{table}[t]
\centering
\caption{Signature and public-key sizes for standardized schemes.}\label{tab:pqc-sizes}
\begin{tabular}{@{}lccc@{}}
\toprule
\textbf{Scheme} & \textbf{NIST Level} & \textbf{Signature} & \textbf{Public key} \\
\midrule
ML-DSA-44~\cite{mldsa}  & 2 & 2,420 bytes & 1,312 bytes \\
ML-DSA-65~\cite{mldsa}  & 3 & 3,309 bytes & 1,952 bytes \\
ML-DSA-87~\cite{mldsa}  & 5 & 4,627 bytes & 2,592 bytes \\
SLH-DSA-128f~\cite{slhdsa} & 1 & 17,088 bytes & 32 bytes \\
FN-DSA-512~\cite{falcon} & 1 & ${\approx}\,666$ bytes & 897 bytes \\
\midrule
\multicolumn{2}{l}{Ed25519 (classical)} & 64 bytes & 32 bytes \\
\multicolumn{2}{l}{secp256k1 ECDSA (classical)} & ${\approx}\,71$ bytes & 33 bytes \\
\bottomrule
\end{tabular}
\end{table}

A ledger that targets a few hundred milliseconds per slot cannot spend that
slot on a kilobyte-scale post-quantum signature at every transfer, and it
cannot spend it on a several-hundred-kilobyte STARK either. ML-DSA-44 is about
2.4\,KB per signature~\cite{mldsa}; SLH-DSA is larger~\cite{slhdsa}.
Ed25519 is 64 bytes. Those numbers are the hot-path objects. They are
verified by a native checker, not by an arithmetic circuit.

The other constraint is the account identifier. If $\ID$ is a signature
public key, rotating Ed25519 to ML-DSA changes the address. A high-throughput
chain that expects a post-quantum migration therefore needs an identifier
that is not welded to one verification algorithm.

\ZKACE occupies a narrow slot in that design. The native account is a hash
of a 256-bit root $\REV$. Account birth (and, optionally, off-slot
aggregation) is an opening of that hash, bound to a registered target, a
typed action digest, a public operation nonce, and a proof-derived replay
value, and completed by one atomic state transition. After birth the account
installs an authentication key. Ordinary transfers verify that key. The
measured reference Circle STARK for the opening is about 339\,KiB: that
artifact is assigned to creation rather than a 400\,ms transfer slot.

\paragraph{Why not prove a signature in zero knowledge.}
One response to large post-quantum signatures is to verify the signature
inside a zero-knowledge circuit and put only the proof on chain~\cite{circom-ecdsa}.
For lattice schemes this is expensive. ML-DSA verification needs number-theoretic
transforms over degree-256 polynomial rings modulo $q=2^{23}-2^{13}+1$,
modular reductions that a proof system over another field must emulate, and
norm and hint checks. Expressed as constraints, this is on the order of
millions of constraints (Section~\ref{sec:circuit:compare}). The approach
keeps the signature-centric model of authorization and moves its cost from
on-chain data to proving. \ZKACE instead proves the smaller statement that the
signature verifier was standing in for.

\subsection{Key Observation}
\label{sec:intro:key-obs}
Consensus needs assurance that a registered identity authorized a specific
action. A signature algorithm is one way to express that assurance. It is
not the assurance itself. Welding
\[
  \text{private key}=\text{identity}=\text{authorization}=\text{signature}
\]
makes the address, the verifier, and the on-chain artifact all change when
the algorithm changes. In-circuit ECDSA or ML-DSA authorization simulates a
signature checker and therefore carries the corresponding circuit cost.
When the secret is already a high-entropy root and the public key is its hash,
the smaller statement is knowledge of the registered opening.

A non-interactive proof of knowledge of that root, with the action in the
statement, follows the Fiat--Shamir / Bellare--Goldwasser proof-of-knowledge
signature pattern~\cite{fiat-shamir,bellare-goldwasser}. The public key is the
on-chain commitment, the signing secret is $\REV$, and the proof is the
authorization artifact. Section~\ref{sec:soundness} proves registered-row
unforgeability of birth and execution with an honestly generated history.
The circuit proves the opening relation, and the account identifier is
independent of ML-DSA, Ed25519, and ECDSA.

The proof therefore is a signature in the proof-of-knowledge sense, and this
paper does not claim authorization without any object of that kind. What
\ZKACE removes from the opening statement is the scheme-specific verifier: its
size and circuit cost are not set by ECDSA, EdDSA, ML-DSA, or SLH-DSA, and the
account identifier survives a change of the installed algorithm.

\subsection{Relation to ACE-GF}
\label{sec:intro:acegf}
ACE-GF~\cite{wang-acegf-2026} separates a high-entropy identity root from
the cryptographic keys and authorization paths derived from it. Keys are
replaceable; the root is not an algorithm. \ZKACE is the on-chain opening of
the native-account hash of that root, plus the executor that consumes a
nonce and a replay value in the same transaction as the call. The hash used
for $\ID$ in this paper is the circuit-native Poseidon2 instance of
Section~\ref{sec:circuit}. Off-chain derivation of installed keys (in ACE-GF,
HKDF~\cite{hkdf} under a context tuple) is outside the circuit; those keys are
verified natively after they have been installed.

\subsection{Contributions}
\label{sec:intro:contrib}
\begin{enumerate}[nosep,leftmargin=1.5em]
  \item A transaction-bound opening relation (C1)--(C5), an enrollment write,
  \textsc{CreateNativeAccount}, and an atomic
  \textsc{VerifyAndExecute} transition
  (Sections~\ref{sec:model}--\ref{sec:construction}).
  \item Game-based unforgeability (birth and execution), substitution
  resistance, and replay resistance relative to the enrollment table
  (Section~\ref{sec:security}).
  \item A specified Poseidon2-M31 Circle STARK profile, an implementation
  artifact snapshot, and
  a ledger placement that uses ZK at account birth and installed keys on the
  transfer hot path
  (Sections~\ref{sec:circuit} and~\ref{sec:implementation}).
  \item A structural comparison of the opening circuit with in-circuit
  verification of post-quantum and classical signatures
  (Section~\ref{sec:circuit:compare}), and a discussion of
  account-abstraction integration (Section~\ref{sec:deployment:aa}).
\end{enumerate}

Enrollment, recovery after loss of $\REV$, and sponsor admission policy
are outside the three games. The measured artifact is a proof engine plus
the authorization relation. Birth is Algorithm~\ref{alg:birth}.

\subsection{Organization}
Section~\ref{sec:related} surveys related work. Section~\ref{sec:model} fixes
notation, registry state, the threat model, and the argument interface.
Section~\ref{sec:construction} states the relation, birth, and the executor.
Section~\ref{sec:circuit} pins encodings and the sponge schedule and compares
the circuit with in-circuit signature verification.
Section~\ref{sec:security} gives the three games.
Section~\ref{sec:implementation} reports the crate, the local evaluation, and
the ledger placement. Section~\ref{sec:limitations} lists limitations and
future work, and Section~\ref{sec:conclusion} concludes.

\section{Related Work}
\label{sec:related}

\paragraph{Signatures from proofs of knowledge.}
Fiat--Shamir identification~\cite{fiat-shamir} and Bellare--Goldwasser
NIZK signatures~\cite{bellare-goldwasser} already treat a proof of knowledge
of a secret, with a message in the statement, as a signature. \ZKACE uses
that pattern with a persistent on-chain commitment as the public key, plus
an enrollment table, a public nonce equality, a canonical action digest, and
an atomic transition. The closest prior art is Picnic~\cite{picnic}, which
turns a proof of knowledge of a block-cipher key into a signature in the
MPC-in-the-head model; its successor FAEST~\cite{faest} does the same in the
VOLE-in-the-head model. There the image is a signature public key tied to one
verification algorithm, whereas here it is a registered account identifier
that outlives the algorithm, and the executor state makes the proof an
authorization. zkLogin~\cite{zklogin} proves possession of an OAuth credential
bound to an ephemeral key, in an issuer-based trust model, and
ERC-4337~\cite{erc4337} leaves the verification algorithm to the account
contract.

\paragraph{Post-quantum signatures on ledgers.}
Direct integration of lattice-based or hash-based signatures into transaction
validation~\cite{mldsa,slhdsa,falcon} is limited by per-transaction data and
verification cost. In \ZKACE, NIST ML-DSA and SLH-DSA~\cite{mldsa,slhdsa} are
the hot-path algorithms after birth, verified natively; the proof does not
replace them on ordinary transfers.

\paragraph{ZK-compressed signature verification.}
Another line of work embeds signature verification inside a circuit:
Circom-based ECDSA verification~\cite{circom-ecdsa}, zkDilithium~\cite{zkdilithium}
for ML-DSA verification in zero knowledge, and lattice-native proof systems
such as LaBRADOR~\cite{labrador}, which reduce but do not remove the cost.
\ZKACE places no signature verification logic inside the circuit
(Section~\ref{sec:circuit:compare}).

\paragraph{Transparent proof systems and the hash instance.}
STARKs and FRI supply transparent arguments whose soundness rests on hash
binding and low-degree tests~\cite{stark,fri}. Poseidon2 is a published
permutation for proof-friendly hashing~\cite{poseidon,poseidon2}. Collision
resistance is claimed for the specified Mersenne-31 instance of
Section~\ref{sec:circuit}. A bit-oriented hash inside the same AIR produces a
multi-megabyte proof and is therefore excluded from the account-creation
profile. The Circle STARK over Mersenne-31~\cite{circlestark}, made
non-interactive by the BCS transformation~\cite{bcs}, is the proof system of
the evaluated profile. Work on zero-knowledge IOPPs with near-optimal
soundness-to-length overhead~\cite{chiesa-zkiopp-2026} improves the proof-system
layer and is complementary to the authorization relation here.

\paragraph{Pairing SNARKs.}
A pairing SNARK such as Groth16~\cite{groth16} can prove
$\mathcal{L}_{\mathsf{ACE}}$ for an existing pairing verifier, which is useful
when a birth proof must hit an EVM precompile; it needs a circuit-specific
reference string, is not post-quantum, and is not measured in this paper.

\paragraph{Commitment and nullifier systems.}
Zerocash~\cite{zcash} and Tornado Cash~\cite{tornado} use commitments with
nullifier sets, and Semaphore~\cite{semaphore} uses identity commitments for
anonymous signalling. \ZKACE shares the consumed-value set but not the
privacy goal: $\ID$ is a one-way image of $\REV$ under a public constant, the
values $\ID$, $\tgt$, $\dom$, $\opnonce$, and $\rp$ are public and linkable,
and the protocol does not provide anonymous membership.

\paragraph{Identity-based cryptography.}
Identity-based encryption~\cite{boneh-franklin} and identity-based
signatures~\cite{ib-sig} bind capabilities to identity strings through a
trusted key-generation centre. \ZKACE has no such party; the identifier is a
hash of user-held root material.

\section{Model}
\label{sec:model}

\subsection{Notation}
\label{sec:notation}

Let $\lambda$ be the security parameter and $\mathcal{F}$ a field. The
security statements below describe a parameterized family
$(\Pi_\lambda,\mathsf{S}_{1,\lambda},H_{\mathsf{msg},\lambda})$; the fixed
Mersenne-31 construction in Section~\ref{sec:circuit} is its concrete
evaluation profile. Let
$\mathsf{Enc}$ be a fixed injective typed encoding into $\mathcal{F}$-elements
or bytes. Concatenation in this paper means that encoding; a parsing
ambiguity voids the corresponding binding claim.

The protocol uses a secret $\REV\in\{0,1\}^{256}$ and a public protocol
constant $\mathsf{s}_{\mathrm{id}}$ (the native-identity salt); a private
context $\ctx=(\mathsf{alg},\mathsf{ctxdom},\mathsf{index})$; public values
$\dom,\tgt,\opnonce,\ID,\msg,\rp$; a sponsor account $\mathsf{sp}$; an
initial verification key $\mathsf{pk}_0$; and an action
\[
\begin{aligned}
  \mathsf{a}=(&\mathsf{executorId},\mathsf{submitterId},
  \mathsf{actionId},\mathsf{resourceId},\\
  &\mathsf{calldata},\mathsf{requestId},
  \mathsf{notBefore},\mathsf{deadline}).
\end{aligned}
\]
The domain $\dom$ is a chain identifier in the backend field
$\mathbb{F}_{2^{31}-1}$. Let $\mathsf{S}_1,\mathsf{S}_{\mathsf{t}},
\mathsf{S}_{\mathsf{r}}$ be three domain-separated Poseidon2 sponges on that
field, with distinct fixed tags in their initial capacities. The evaluated
profile instantiates $H_{\mathsf{msg}}$ as SHA-256 computed \emph{outside} the
circuit over the byte encoding specified in Section~\ref{sec:action}.

\subsection{Registry and state}
\label{sec:registry}

An authenticated registry $\mathcal{R}$ and mutable state
$(\mathcal{N},\mathcal{U},\mathcal{K},\mathcal{C})$ are part of the
executor. Four roles touch them: the holder of $\REV$, who proves; the
sponsor, an already-created account that pays fees at birth; the ledger,
which verifies and writes; and the application, which names
$\mathsf{a}$ and dispatches later calls. The proof engine implements
none of the writes.
\begin{align*}
  \mathcal{R}[\dom,\ID] &\subseteq \{\tgt\},\\
  \mathcal{N}[\dom,\ID,\tgt] &\in \{0,\ldots,2^{64}-1\},\\
  \mathcal{U} &\subseteq \{\rp\},\\
  \mathcal{K}[\dom,\ID] &\subseteq \{\mathsf{pk}\},\\
  \mathcal{C} &\subseteq \{\ID\}.
\end{align*}
$\mathcal{C}$ is the set of created native accounts. $\mathcal{K}$ is the
installed-key table used on the transfer hot path.
$\mathcal{N}$ is the \ZKACE operation counter; it is not the installed-key
sequence number. $\ID$ is missing from $\mathcal{C}$ until
\textsc{CreateNativeAccount} commits. A transfer to an absent $\ID$
rejects; the ledger does not create a zero-key balance account.

Those writes are produced by Algorithms~\ref{alg:birth}
and~\ref{alg:protocol}. Theorems treat a committed row as authentic.
Without a committed row, a verifier can conclude only that some opening
was known, not that a particular account exists or that a particular key
is installed.

\subsection{Threat model}
\label{sec:threat}
The adversary sees all ledger state, an honestly generated history of
accepted proofs, and the mempool. It can
inject, delay, replay, and reorder transactions, and it can submit
arbitrary public inputs and proof bytes. It can try to bind a valid proof
to another action, domain, target, nonce, executor, authenticated submitter,
request identifier, time window, or calldata, and it can race two
executors on the same $\rp$ or the same $(\dom,\ID,\tgt,\opnonce)$. It
does not receive $\REV$. Compromise of $\REV$ is key compromise and is
out of scope.

Assets are the confidentiality of $\REV$ and $\ctx$; the integrity of
$(\dom,\ID,\tgt)$ and of $\mathsf{a}$; exactly-once execution of
$(\mathcal{N},\mathcal{U})$; the installed-key table $\mathcal{K}$; and
the pinned verifying key of $\Pi$. Public $\ID$, $\tgt$, $\dom$,
$\opnonce$, and $\rp$ are linkable. The protocol does not provide
anonymous membership.

Table~\ref{tab:threats} is the in-scope map. Out of scope: who is allowed
to act as sponsor (faucet policy, invitation policy); account recovery
after loss of $\REV$; governance that later adds or revokes installed
keys beyond Algorithm~\ref{alg:birth}; and a recursive block verifier.

\begin{table}[t]
\centering
\caption{In-scope threats and the object that stops them.}\label{tab:threats}
\small
\begin{tabular}{@{}p{0.46\linewidth}p{0.46\linewidth}@{}}
\toprule
\textbf{Attempt} & \textbf{Stop} \\
\midrule
Proof for an unregistered $\ID$ or $\tgt$
  & knowledge soundness plus $\mathcal{R}$ lookup \\
Birth of an $\ID$ the adversary does not open
  & knowledge extraction of an opening (Thm.~\ref{thm:birthuf}) \\
Empty account, then hijack $\mathsf{pk}_0$
  & no account exists until the atomic birth write \\
Swap $\mathsf{a}$, $\dom$, $\tgt$, or $\opnonce$
  & public-input binding and $H_{\mathsf{msg}}$ \\
Submit through another principal or executor
  & authenticated submitter check and configured executor check \\
Reuse $\rp$ or skip $\opnonce$
  & serializable $\mathcal{N}$ and $\mathcal{U}$ \\
Cross-chain replay
  & $\dom$ in (C1), (C3), $x$, $H_{\mathsf{msg}}$, and $\mathcal{R}$ \\
Early or expired call
  & trusted executor clock against $\mathsf{notBefore}$ and $\mathsf{deadline}$ \\
Sponsor substitution of $\ID$ or $\mathsf{pk}_0$
  & $\msg$ binds payload and sponsor; proof is on that $\msg$ \\
\bottomrule
\end{tabular}
\end{table}

Protocol acceptance requires verification and the state write to commit in
one transaction. Stateless verification and HTTP idempotency do not replace
the consumed-value set $\mathcal{U}$.

\paragraph{Where user intent enters.}
In an ACE-GF wallet the prover obtains $\REV$ by unsealing a sealed artifact
with a credential~\cite{wang-acegf-2026}, so proof generation is gated on
credential presentation. That gate is not per-action approval: while $\REV$ is
in prover memory, a compromised wallet or host can prove arbitrary actions.
Confirmation dialogs, spending limits, and unseal-session lifetime are wallet
policy and are outside this paper.

\subsection{Argument system}
\label{sec:proof-interface}

Let $\Pi=(\mathsf{Setup},\mathsf{Prove},\mathsf{Verify})$ be a non-interactive
argument for a relation $\mathcal{L}$. Completeness and knowledge soundness
are required for unforgeability, where the auxiliary input is the public
state and the honest history of Section~\ref{sec:soundness}; in the STARK
instantiation this is straight-line extraction in the random-oracle model.
Knowledge soundness for arbitrary auxiliary input is not assumed: it is known
to be unattainable for general SNARKs under indistinguishability
obfuscation~\cite{bcpr}. Public-input binding is required for
substitution resistance: no \PPT{} adversary outputs $(\pi,x,x')$ with
$x\neq x'$ and $\mathsf{Verify}(\vk,\pi,x)=\mathsf{Verify}(\vk,\pi,x')=1$
except with negligible probability. Zero knowledge is required only for the
privacy claim in Section~\ref{sec:privacy}. The relation does not depend on
the pairing-based or hash-based realisation of $\Pi$. Setup returns public
parameters $\pp$ containing the verification key $\vk$; algorithms below
write verification uniformly as $\mathsf{Verify}(\vk,\pi,x)$.

\section{Construction}
\label{sec:construction}

\subsection{Relation}
\label{sec:relation}

For public input
\[
  x=(\ID,\msg,\dom,\opnonce,\tgt,\rp)
\]
and witness
\[
  w=(\REV,\mathsf{alg},\mathsf{ctxdom},\mathsf{index},\mathsf{n}),
\]
the relation $\mathcal{L}_{\mathsf{ACE}}(x,w)=1$ if and only if
\begin{align}
  \ID &=
    \mathsf{S}_1\bigl(\mathsf{Enc}(\REV,\mathsf{s}_{\mathrm{id}},\dom)\bigr), \tag{C1}\\
  \tgt &=
    \mathsf{S}_{\mathsf{t}}\bigl(\mathsf{Enc}(\REV,\mathsf{alg},
      \mathsf{ctxdom},\mathsf{index})\bigr), \tag{C2}\\
  \rp &=
    \mathsf{S}_{\mathsf{r}}\bigl(\mathsf{Enc}(\REV,\mathsf{alg},\mathsf{ctxdom},
      \mathsf{index},\msg,\dom,\mathsf{n})\bigr), \tag{C3}\\
  \mathsf{ctxdom} &= \dom, \tag{C4}\\
  \mathsf{n} &= \opnonce. \tag{C5}
\end{align}
(C4) is equality of one field element. The native-identity salt
$\mathsf{s}_{\mathrm{id}}$ is a fixed public constant of the relation, not a
witness value. (C5) is equality of the full 64-bit
nonce; in the STARK profile a $\mathsf{u64}$ occupies three Mersenne-31
limbs and every limb is constrained. Omitting $\opnonce$ from $x$, or
constraining only the low limb of the nonce, is not this relation: the
executor would then be free to consume a counter that the proof did not bind.

\subsection{Action digest}
\label{sec:action}

\begin{equation}
\begin{aligned}
\msg = \mathsf{SHA256}\bigl(
 &\texttt{ZK-ACE/AUTH}\parallel\ID\parallel\dom\parallel\opnonce\parallel\tgt\\
 &\parallel\mathsf{executorId}\parallel\mathsf{submitterId}
 \parallel\mathsf{actionId}\parallel\mathsf{resourceId}\\
 &\parallel\mathsf{LE64}(|\mathsf{calldata}|)
 \parallel\mathsf{SHA256}(\mathsf{calldata})\\
 &\parallel\mathsf{LE64}(|\mathsf{requestId}|)
 \parallel\mathsf{requestId}\\
 &\parallel\mathsf{LE64}(\mathsf{notBefore})
 \parallel\mathsf{LE64}(\mathsf{deadline})\bigr).
\end{aligned}
\label{eq:msg}
\end{equation}
All Mersenne-31 limbs are encoded as canonical little-endian 32-bit words;
$\mathsf{executorId}$, $\mathsf{submitterId}$, $\mathsf{actionId}$, and
$\mathsf{resourceId}$ are fixed 32-byte strings, so the encoding is
injective without length prefixes for them; $\mathsf{calldata}$ and
$\mathsf{requestId}$ are length-prefixed, and $\mathsf{requestId}$ is its UTF-8
byte sequence. The label
$\texttt{ZK-ACE/AUTH}$ separates this digest from other hashes of the same
tuple shape. Because $\ID$, $\dom$, $\opnonce$, and $\tgt$ are inputs
of~\eqref{eq:msg}, moving a proof to another registered binding or sequence slot
requires a collision in $H_{\mathsf{msg}}$ unless the public statement
changes (Theorem~\ref{thm:substitution}). The executor must check
$\mathsf{executorId}$ against its configured identity,
$\mathsf{submitterId}$ against the transport-authenticated principal, and the
time window against its trusted clock. It must dispatch the exact committed
$\mathsf{calldata}$. The hash does not enforce those facts if the executor
ignores them. The reference handler caps calldata at 1\,MiB and request
identifiers at 128 UTF-8 bytes.

\subsection{Birth and enrollment}
\label{sec:birth}

\textsc{EnrollBindingInternal} is an internal registry write, callable only
inside an authenticated ledger transition such as
\textsc{CreateNativeAccount}. A separately exposed target-governance
operation requires installed-key authorization and is outside this paper.
The internal write does not create an account or install a key.
\textsc{CreateNativeAccount} is the
only native birth: a holder of $\REV$ proves an opening, a sponsor pays,
and the ledger creates $\ID$, enrolls $\tgt$, installs $\mathsf{pk}_0$,
and consumes $\rp$ in one snapshot. The sender of the transaction is the
sponsor, not the newborn. $\ID$ in the payload is the new identity.

The birth digest is~\eqref{eq:msg} on
\[
\begin{aligned}
  \mathsf{a}_{\mathrm{birth}}=
  (&\mathsf{ledgerId},\mathsf{sp},\texttt{CreateNativeAccount},\ID,\\
  &\mathsf{Enc}(\mathsf{pk}_0),\mathsf{requestId},
  \mathsf{notBefore},\mathsf{deadline}),
\end{aligned}
\]
together with $\opnonce=0$. Its components fill the fields of $\mathsf{a}$ in
order: $\mathsf{executorId}=\mathsf{ledgerId}$, $\mathsf{submitterId}=\mathsf{sp}$,
$\mathsf{actionId}$ is the fixed 32-byte identifier of
\textsc{CreateNativeAccount}, $\mathsf{resourceId}=\ID$ (as its 32-byte limb
encoding), and $\mathsf{calldata}=\mathsf{Enc}(\mathsf{pk}_0)$. The hash input is the payload and the
sponsor identity. It does not include the proof bytes: putting $\rp$ or
$\pi$ into $\msg$ makes the statement circular. The circuit does not
check that $\mathsf{pk}_0$ is an HKDF image of $\REV$. The holder of
$\REV$ names $\mathsf{pk}_0$ in the payload; $H_{\mathsf{msg}}$ binds
that choice. For a fixed prospective $\ID$, an accepted birth proof without
an opening contradicts the extraction or identity-hash assumptions used in
Section~\ref{sec:soundness}. The
sponsor cannot replace $\ID$ or $\mathsf{pk}_0$ without changing $\msg$.

No account in $\mathcal{C}$ has an empty $\mathcal{K}$.
$\mathsf{SetAuthPubkey}$ after birth is an installed-key operation on an
element of $\mathcal{C}$. It is not a first-claim path.
Authenticated native-account state additionally maintains that
$\ID\notin\mathcal{C}$ implies
$\mathcal{R}[\dom,\ID]=\emptyset$ and
$\mathcal{K}[\dom,\ID]=\emptyset$, with
$\mathcal{N}[\dom,\ID,\cdot]$ undefined. Thus an uncreated identity cannot
carry a preinstalled target, nonce, or authorization key. The cost is that
$\ID$ cannot receive funds before its birth.

\begin{algorithm}[t]
\caption{Native-account birth}
\label{alg:birth}
\begin{algorithmic}[1]
\Require public parameters $\pp$, authenticated state
$(\mathcal{R},\mathcal{N},\mathcal{U},\mathcal{K},\mathcal{C})$
\Procedure{EnrollBindingInternal}{$\dom,\ID,\tgt,\opnonce_0$}
  \State reject if $\opnonce_0=2^{64}-1$
  \State reject if $\tgt\in\mathcal{R}[\dom,\ID]$ and
      $\mathcal{N}[\dom,\ID,\tgt]\neq\opnonce_0$
  \State $\mathcal{R}[\dom,\ID]\gets\mathcal{R}[\dom,\ID]\cup\{\tgt\}$
  \State $\mathcal{N}[\dom,\ID,\tgt]\gets\opnonce_0$
\EndProcedure
\Procedure{CreateNativeAccount}{$\mathsf{sp},\ID,\mathsf{pk}_0,x,\pi,
    \mathsf{a}_{\mathrm{birth}}$}
  \State reject unless the transaction sender is $\mathsf{sp}$
      and $\mathsf{sp}\in\mathcal{C}$
  \State reject unless $\ID\notin\mathcal{C}$ and $\ID\neq\mathsf{sp}$
  \State reject unless $x=(\ID,\msg,\dom,0,\tgt,\rp)$ and $\msg$
      equals~\eqref{eq:msg} for $(\ID,\dom,0,\tgt,\mathsf{a}_{\mathrm{birth}})$
  \State reject unless $\mathsf{a}_{\mathrm{birth}}$ names $\mathsf{sp}$
      and $\mathsf{pk}_0$, is submitted by its committed sponsor to its
      committed executor, and is within its committed time window
  \State reject unless $\mathsf{Verify}(\vk,\pi,x)=1$
  \State reject unless $\mathsf{sp}$ pays fees and any initial credit
      under the ledger's installed-key rule
  \State \textbf{atomic snapshot:} require $\ID\notin\mathcal{C}$ and
      $\rp\notin\mathcal{U}$;
      \Call{EnrollBindingInternal}{$\dom,\ID,\tgt,0$};
      $\mathcal{C}\gets\mathcal{C}\cup\{\ID\}$;
      $\mathcal{K}[\dom,\ID]\gets\{\mathsf{pk}_0\}$;
      $\mathcal{N}[\dom,\ID,\tgt]\gets 1$;
      $\mathcal{U}\gets\mathcal{U}\cup\{\rp\}$;
      credit $\ID$ from $\mathsf{sp}$ if named
  \State on any failure, abort the snapshot: $\ID\notin\mathcal{C}$ and
      $\rp$ is not recorded
\EndProcedure
\end{algorithmic}
\end{algorithm}

\subsection{Authorization after birth}
\label{sec:algorithms}

\begin{algorithm}[t]
\caption{\ZKACE authorization}
\label{alg:protocol}
\begin{algorithmic}[1]
\Require public parameters $\pp$, authenticated state
$(\mathcal{R},\mathcal{N},\mathcal{U},\mathcal{K},\mathcal{C})$
\Procedure{ProveAction}{$\REV,\ctx,\mathsf{n},\dom,\mathsf{a}$}
  \State compute $(\ID,\tgt)$ from (C1)--(C2)
  \State $\opnonce\gets\mathsf{n}$; compute $\msg$ by~\eqref{eq:msg} and $\rp$ by (C3)
  \State $x\gets(\ID,\msg,\dom,\opnonce,\tgt,\rp)$
  \State \Return $(x,\mathsf{Prove}(\pp,x,w))$
\EndProcedure
\Procedure{VerifyAndExecute}{$x,\pi,\mathsf{a}$}
  \State reject unless $\ID\in\mathcal{C}$
  \State reject unless $\msg$ equals~\eqref{eq:msg} for
      $(\ID,\dom,\opnonce,\tgt,\mathsf{a})$
  \State reject unless $\mathsf{Verify}(\vk,\pi,x)=1$
  \State reject if the executor or authenticated submitter differs, the
      action is outside its time window, forbidden, or $\opnonce=2^{64}-1$
  \State atomically require $\tgt\in\mathcal{R}[\dom,\ID]$,
      $\opnonce=\mathcal{N}[\dom,\ID,\tgt]$, and $\rp\notin\mathcal{U}$
  \State atomically set
    $\mathcal{N}[\dom,\ID,\tgt]\gets\opnonce+1$ and
    $\mathcal{U}\gets\mathcal{U}\cup\{\rp\}$, then execute $\mathsf{a}$
\EndProcedure
\end{algorithmic}
\end{algorithm}

The atomic require and set in \textsc{VerifyAndExecute} occur in the same
transaction as the call. A check followed by a later write is not this
protocol. After birth, ordinary transfers verify a key in
$\mathcal{K}[\dom,\ID]$ and do not run this procedure.

\begin{theorem}[Birth completeness]
\label{thm:birth}
If $\Pi$ is complete, $\ID\notin\mathcal{C}$, $\mathsf{sp}\in\mathcal{C}$,
the sponsor can pay, $\rp\notin\mathcal{U}$, and an honest proof of
$\mathcal{L}_{\mathsf{ACE}}$ is
supplied for $\mathsf{a}_{\mathrm{birth}}$ at $\opnonce=0$, then
\textsc{CreateNativeAccount} accepts, $\ID\in\mathcal{C}$,
$\mathcal{K}[\dom,\ID]=\{\mathsf{pk}_0\}$, $\tgt\in\mathcal{R}[\dom,\ID]$,
$\mathcal{N}[\dom,\ID,\tgt]=1$, and $\rp\in\mathcal{U}$. Failure of any
check leaves $\ID\notin\mathcal{C}$.
\end{theorem}
\begin{proof}
(C1)--(C5) hold by construction at nonce $0$, so verification accepts.
Equation~\eqref{eq:msg} is recomputed from the payload and sponsor. The
snapshot writes $\mathcal{C}$, $\mathcal{K}$, $\mathcal{R}$, $\mathcal{N}$,
and $\mathcal{U}$ together.
\end{proof}

\begin{theorem}[Completeness]
\label{thm:completeness}
If $\Pi$ is complete, $\ID\in\mathcal{C}$, $\tgt\in\mathcal{R}[\dom,\ID]$,
$\mathcal{N}[\dom,\ID,\tgt]=\opnonce<2^{64}-1$, $\rp\notin\mathcal{U}$,
and the action is allowed and unexpired, then an
honest \textsc{ProveAction} followed by \textsc{VerifyAndExecute} accepts and
performs exactly one state transition.
\end{theorem}
\begin{proof}
(C1)--(C5) hold by construction, so verification accepts. Equation~\eqref{eq:msg} is
recomputed from the same typed values. The registry and freshness predicates
hold by hypothesis. The transition inserts $\rp$ and increments the matched
nonce once.
\end{proof}

\subsection{What the relation does not prove}
\label{sec:not-prove}
A valid proof of $\mathcal{L}_{\mathsf{ACE}}$ says that some opening of the
declared $\ID$ produced the declared $(\tgt,\msg,\dom,\opnonce,\rp)$. It does
not say that $\mathsf{pk}_0$ is a KDF image of $\REV$, that the sponsor is
an honest faucet, that the time window is current, or that the executor will
dispatch the exact committed calldata.
\textsc{CreateNativeAccount} and \textsc{VerifyAndExecute} are the executor
obligations that turn the opening into an account or a state transition
(Section~\ref{sec:registry}). The relation also does not verify an Ed25519
or ML-DSA signature; after birth those checks are native on $\mathcal{K}$.
No signature object enters the witness, and the relation is not a compression
of a post-quantum signature.

\section{Circuit profile}
\label{sec:circuit}

This section specifies the STARK conformance target for
$\mathcal{L}_{\mathsf{ACE}}$ in the \texttt{zk-ace} crate. A conforming
implementation must enforce the same encodings, sponge schedule, and public
bindings in native hashing, AIR, and verification.

\subsection{Field encodings}
\label{sec:circuit:enc}
The backend field is $\mathbb{F}_{2^{31}-1}$. A 32-byte string occupies nine
limbs: eight 31-bit chunks and one 8-bit overflow limb in
$\{0,\ldots,255\}$. This range restriction is part of $\mathsf{Enc}$ and is
a conformance requirement for the AIR and public-input decoder; without it
the field representation is not injective. A 64-bit nonce
occupies canonical $30/30/4$-bit limbs $(n_0,n_1,n_2)$ with
$n_0,n_1<2^{30}$ and $n_2<2^4$;
every limb is constrained (C5). Domain is one field element (C4): a chain
identifier that fits in Mersenne-31. Stretching C4 to a lossless $\mathsf{u64}$
would add limb constraints without changing the product statement.

Public inputs are 37 field elements in this order:
$\ID[8]$, $\msg[9]$, $\dom[1]$, $\opnonce[3]$, $\tgt[8]$, $\rp[8]$
(148 bytes). Poseidon2 outputs are eight rate elements. Native-account salt
is the public protocol constant $\mathsf{s}_{\mathrm{id}}$, a 32-byte string
occupying nine limbs; $\ID$ is then an
OWF image of $\REV$, not a hiding commitment.

\subsection{Specified Poseidon2 instance}
\label{sec:circuit:p2}
Table~\ref{tab:poseidon2} is the permutation. Round constants are the
Grain-LFSR stream of the Plonky3 Mersenne-31 generator~\cite{plonky3-m31}
(\texttt{generate\_constants.py --field mersenne31 --width 16}). The external
linear layer is the published Poseidon2 $\mathsf{MDSMat4}$ light matrix; the
internal diagonal is the published Mersenne-31 width-16 vector
$[-2,1,2,4,8,16,32,64,2^{7},2^{8},2^{10},2^{12},2^{13},2^{14},2^{15},2^{16}]$.
The collision-resistance assumption in Section~\ref{sec:security} applies to
this exact tuple $(p,t,\alpha,R_F,R_P,\mathsf{MDS},\mathsf{RC})$.

\begin{table}[t]
\centering
\caption{Specified Poseidon2-M31 instance.}\label{tab:poseidon2}
\small
\begin{tabular}{@{}p{0.22\linewidth}p{0.70\linewidth}@{}}
\toprule
Parameter & Value \\
\midrule
Field $p$ & $2^{31}-1$ \\
Width $t$ & $16$ (capacity $8$, rate $8$) \\
S-box $\alpha$ & $5$ \\
$R_F+R_P$ & $8+14$ \\
Linear layers & Poseidon2 $\mathsf{MDSMat4}$ light; published internal diagonal \\
Round constants & Grain-LFSR, Plonky3 Mersenne-31 generator \\
\bottomrule
\end{tabular}
\end{table}

Each sponge writes the input length into $\mathsf{capacity}[0]$, writes the
fixed relation tag \texttt{ID\_COM}, \texttt{DERIVE}, or \texttt{AUTH} into
$\mathsf{capacity}[1]$, and absorbs rate-8 blocks. After the last partial
block the permutation runs once more.
$\mathsf{S}_{\mathsf{t}}$ and $\mathsf{S}_{\mathsf{r}}$ squeeze a second
rate-8 output with one extra permutation; the public values $\tgt$ and $\rp$
are those second lanes. The absorbed tuples remain those of (C2) and (C3).

\subsection{Sponge schedule}
\label{sec:circuit:sched}
C1 absorbs 19 elements ($\REV$, $\mathsf{s}_{\mathrm{id}}$, $\dom$): three permutations.
C2 absorbs 12 elements ($\REV$, context): two permutations plus one squeeze.
C3 absorbs 25 elements ($\REV$, context, $\msg$, $\dom$, nonce): four
permutations plus one squeeze. The AIR is eleven permutations of
$22$ rounds plus an output row, $11\times 23=253$ active rows, inside a
length-$512$ trace. Unused rows are filled with fresh field elements.
Witness constancy is restricted to active rows. Selectors and round
constants are preprocessed; the verifier pins their Merkle root.

The first permutation of each sponge starts from the length-tagged absorbed
state after the external linear layer. Continuation rows add the next rate
block and apply that layer again. Full rounds evaluate $x\mapsto (x+c)^5$
on all sixteen coordinates and the external layer; partial rounds evaluate
the S-box on coordinate $0$ and the internal layer.

\subsection{Comparison with in-circuit signature verification}
\label{sec:circuit:compare}
The opening circuit is eleven Poseidon2 permutations over a short witness. A
circuit that verifies a signature instead must express the verifier of that
scheme. For ML-DSA this means (i)~number-theoretic transforms over degree-256
polynomial rings in $\mathbb{Z}_q$ with $q=2^{23}-2^{13}+1$, which a proof
system over another field must emulate with range checks per limb;
(ii)~the $k\times\ell$ matrix--vector product ($4\times 4$ at ML-DSA-44); and
(iii)~hint reconstruction and norm checks. For SLH-DSA it means the
WOTS$^+$ hash chains and Merkle authentication paths of the hypertree, and
for ECDSA it means non-native scalar multiplication and modular inversion.
Table~\ref{tab:constraint-comparison} records the dominant cost of each
circuit. The signature rows are structural order-of-magnitude estimates for
R1CS circuits over a ${\approx}\,$254-bit field, not measurements of this
work; the \ZKACE row is the specified AIR of Section~\ref{sec:circuit:sched}.
AIR rows and R1CS constraints are different units, so the table supports an
order-of-magnitude comparison and no single ratio is claimed.

\begin{table}[t]
\centering
\caption{Opening circuit versus in-circuit signature verification. Signature rows are structural estimates; the ECDSA row is the published circom-ecdsa figure.}\label{tab:constraint-comparison}
\footnotesize
\begin{tabular}{@{}lll@{}}
\toprule
\textbf{Circuit} & \textbf{Size} & \textbf{Dominant cost} \\
\midrule
\ZKACE opening (this work) & 253 active AIR rows, 512-row trace & 11 Poseidon2-M31 permutations \\
\midrule
In-circuit ML-DSA-44 verify & ${\gtrsim}\,$2\,M R1CS (estimate) & $4\!\times\!4$ NTTs, non-native mod.\ arith. \\
In-circuit ML-DSA-65 verify & ${\gtrsim}\,$4\,M R1CS (estimate) & $6\!\times\!5$ NTTs, non-native mod.\ arith. \\
In-circuit FN-DSA-512 verify & ${\gtrsim}\,$1\,M R1CS (estimate) & NTT multiplication, norm check \\
In-circuit SLH-DSA verify & ${\gtrsim}\,$5\,M R1CS (estimate) & WOTS$^+$ chains, Merkle paths \\
In-circuit ECDSA verify~\cite{circom-ecdsa} & ${\approx}\,$1.5\,M R1CS & scalar mul., mod.\ inverse \\
\bottomrule
\end{tabular}
\end{table}

The opening circuit is small because it proves knowledge of the registered
opening, which is the statement a signature verifier stands in for when the
public key is a hash of a high-entropy root. The proof object is not small:
the measured Circle STARK is about 339\,KiB (Section~\ref{sec:not-reported}),
which is why the ledger placement of Section~\ref{sec:deployment} puts it at
account birth and not on the transfer hot path.

\section{Security}
\label{sec:security}

All games are classical and relative to a committed
\textsc{CreateNativeAccount} row and a serializable atomic executor. They
do not speak about sponsor admission policy, about hiding public
identifiers, or about an application that maps committed $\mathsf{calldata}$ to
the wrong effect. Table~\ref{tab:assumptions}
records the reduction chain.

The fixed Mersenne-31 profile is not assigned a concrete classical or
post-quantum security level in this work. Its eight-element output has
slightly less than 248 bits of range entropy, which is only an output-space
ceiling; with a 248-bit capacity, generic collision search on the sponge costs
about $2^{124}$ classical operations. Neither number is a security proof for
the instantiated permutation, AIR, Fiat--Shamir transform, or PCS. Post-quantum installed keys and the concrete
security of this proof profile are separate claims.

\begin{table}[t]
\centering
\caption{Assumptions consumed by each claim.}\label{tab:assumptions}
\small
\begin{tabular}{@{}p{0.22\linewidth}p{0.70\linewidth}@{}}
\toprule
\textbf{Claim} & \textbf{Assumptions} \\
\midrule
Thm.~\ref{thm:birth} & completeness of $\Pi$ \\
Thm.~\ref{thm:completeness} & completeness of $\Pi$; $\ID\in\mathcal{C}$ \\
Thm.~\ref{thm:unforge} & auxiliary-input knowledge soundness of $\Pi$;\newline
  CR of $\mathsf{S}_1$; recovery of $\REV$ from the public history
  (zero knowledge of $\Pi$) \\
Thm.~\ref{thm:birthuf} & as Thm.~\ref{thm:unforge}; CR of $H_{\mathsf{msg}}$ \\
Thm.~\ref{thm:substitution} & public-input binding of $\Pi$;\newline
  CR of $H_{\mathsf{msg}}$ \\
Thm.~\ref{thm:replay} & serializable atomic writes; Thm.~\ref{thm:unforge} \\
\bottomrule
\end{tabular}
\end{table}

\subsection{Unforgeability}
\label{sec:soundness}

\begin{definition}[Registered-row unforgeability with public history]
\label{def:unforge}
$\mathsf{Exp}^{\mathsf{uf}}_{\mathcal{A}}(\lambda)$:
\begin{enumerate}[nosep,leftmargin=1.5em]
  \item \textbf{Setup.} Challenger samples $\REV\xleftarrow{\$}\{0,1\}^{256}$
  and a derivation context producing $\tgt$ via (C2). It sets
  $\ID=\mathsf{S}_1(\mathsf{Enc}(\REV,\mathsf{s}_{\mathrm{id}},\dom))$, writes
  a complete honest birth post-state
  \[
    \ID\in\mathcal{C},\quad
    \mathcal{R}[\dom,\ID]=\{\tgt\},\quad
    \mathcal{N}[\dom,\ID,\tgt]=1,\quad
    \mathcal{K}[\dom,\ID]=\{\mathsf{pk}_0\},\quad
    \rp_{\mathrm{birth}}\in\mathcal{U}.
  \]
  For any fixed polynomial $q=q(\lambda)$, it then generates an honest history
  $\mathcal{H}$ of $q$ successful, consumed authorizations and advances
  $(\mathcal{N},\mathcal{U})$
  accordingly. It gives $\mathcal{A}$ the verification key, complete public
  state, and $\mathcal{H}$, but not $\REV$ or any witness.
  \item \textbf{Output.} $\mathcal{A}$ returns $(x^*,\pi^*,\mathsf{a}^*)$.
  \item \textbf{Win.} \textsc{VerifyAndExecute}$(x^*,\pi^*,\mathsf{a}^*)$
  accepts and $(x^*,\pi^*,\mathsf{a}^*)$ is not an element of $\mathcal{H}$.
\end{enumerate}
$\Adv^{\mathsf{uf}}_{\mathsf{ACE}}(\mathcal{A})=
\Pr[\mathsf{Exp}^{\mathsf{uf}}_{\mathcal{A}}=1]$.
\end{definition}

The winning condition is an accepted transition for the enrolled commitment,
not ``extractor failure''.

\begin{definition}[Recovery]
\label{def:rec}
Challenger samples $\REV$ and a context as in
Definition~\ref{def:unforge}, publishes the complete public birth post-state
and an honestly generated polynomial history $\mathcal{H}$ with the
corresponding post-state (including, for Theorem~\ref{thm:birthuf}, a pending
honest birth tuple), and gives no witness. $\mathcal{A}$ wins
$\mathsf{Exp}^{\mathsf{rec}}$ by outputting $\REV$.
\end{definition}

Without proofs in $\mathcal{H}$, recovery of a uniform 256-bit root is a
multi-image preimage problem for $\mathsf{S}_1$, $\mathsf{S}_{\mathsf{t}}$,
and $\mathsf{S}_{\mathsf{r}}$ on the published $\ID$, $\tgt$, and $\rp$ values.
The proofs in $\mathcal{H}$ add to the adversary's view. If $\Pi$ is
zero-knowledge they can be replaced by simulated proofs, and
$\Adv^{\mathsf{rec}}$ is at most the preimage advantage plus $q+1$ times the
zero-knowledge advantage of $\Pi$. For a backend whose zero knowledge is not
established, such as the prototype of Section~\ref{sec:prototype},
$\Adv^{\mathsf{rec}}$ is an assumption on the proof system and is recorded as
such in Table~\ref{tab:assumptions}. A low-entropy $\REV$ falls to offline
search against the public $\ID$ and is outside the experiment.

\begin{theorem}[Unforgeability]
\label{thm:unforge}
If $\Pi$ is auxiliary-input knowledge-sound for
$\mathcal{L}_{\mathsf{ACE}}$ (for the auxiliary input of
Definition~\ref{def:unforge}) and
$\mathsf{S}_1$ is collision-resistant, then for every \PPT{} $\mathcal{A}$
there are reductions $\mathcal{B}_1,\mathcal{B}_2,\mathcal{B}_3$ with
\[
  \Adv^{\mathsf{uf}}_{\mathsf{ACE}}(\mathcal{A})
  \le
  \Adv^{\mathsf{ks}}_{\Pi}(\mathcal{B}_1)
  +\Adv^{\mathsf{coll}}_{\mathsf{S}_1}(\mathcal{B}_2)
  +\Adv^{\mathsf{rec}}(\mathcal{B}_3).
\]
On every accepting execution the extracted witness satisfies (C1)--(C5).
Either it is a second opening of the registered $\ID$, or it is the enrolled
secret.
\end{theorem}
\begin{proof}
Let Game$_0$ be Definition~\ref{def:unforge} and let $S_i$ be the event that
$\mathcal{A}$ wins in Game$_i$.

\emph{Game$_1$ (extraction).} Run the knowledge extractor $\mathcal{E}$ on an
accepting $(x^*,\pi^*)$. Abort (count as a loss) if $\mathcal{E}$ fails to
return $w^*$ with $\mathcal{L}_{\mathsf{ACE}}(x^*,w^*)=1$. Reduction
$\mathcal{B}_1$ is the malicious prover that simulates Setup for $\mathcal{A}$
(it samples $\REV$ itself) and forwards $\mathcal{A}$'s proof to
the knowledge-soundness challenger. Then
$|\Pr[S_0]-\Pr[S_1]|\le\Adv^{\mathsf{ks}}_{\Pi}(\mathcal{B}_1)$. Conditioned on
not aborting,
\[
  \mathsf{S}_1(\mathsf{Enc}(\REV^*,\mathsf{s}_{\mathrm{id}},\dom^*))=\ID^*=\ID
\]
by (C1) and the executor's check that $\ID^*$ is the registered commitment.
Also $\dom^*=\dom$ by (C4) and the registry key.

\emph{Game$_2$ (no second opening).} Abort on
$\mathsf{Coll}=\{\REV^*\neq\REV\}$.
Reduction $\mathcal{B}_2$ samples the honest opening, runs Game$_1$, and on
$\mathsf{Coll}$ outputs the two distinct $\mathsf{S}_1$-preimages. Hence
$\Pr[\mathsf{Coll}]\le\Adv^{\mathsf{coll}}_{\mathsf{S}_1}(\mathcal{B}_2)$.

\emph{Bounding Game$_2$.} A win now implies $\REV^*=\REV$. Reduction
$\mathcal{B}_3$ plays $\mathsf{Exp}^{\mathsf{rec}}$: it receives
the public post-state and $\mathcal{H}$, starts $\mathcal{A}$ on that view,
and on an accepting
proof runs $\mathcal{E}$ and outputs the extracted $\REV^*$. It never needs
$\REV$ to simulate the public-history game. The view is identical to Game$_2$, so
$\Pr[S_2]\le\Adv^{\mathsf{rec}}(\mathcal{B}_3)$.

Summing the hops gives the bound. Hitting the registered $\tgt$ through (C2)
with a different opening is already $\mathsf{Coll}$.
\end{proof}

The experiment covers an honestly generated public history but does not give
$\mathcal{A}$ an adaptive chosen-action proving oracle. Consequently the
theorem is narrower than EUF-CMA security for a general-purpose signature
scheme.

\begin{definition}[Birth unforgeability]
\label{def:birthuf}
Challenger samples $\REV$, $\dom$, and a context, computes $(\ID,\tgt)$ by
(C1)--(C2), and fixes a sponsor $\mathsf{sp}\in\mathcal{C}$ with
$\ID\notin\mathcal{C}$. An honest holder produces a birth tuple
$(x,\pi,\mathsf{a}_{\mathrm{birth}})$ for some $\mathsf{pk}_0$; it is pending
and not yet committed. $\mathcal{A}$ receives $\vk$, the public state, and the
pending tuple, but not $\REV$. $\mathcal{A}$ outputs
$(x^*,\pi^*,\mathsf{a}^*_{\mathrm{birth}})$ and wins if
\textsc{CreateNativeAccount} accepts it with identity $\ID$ and
$\mathsf{a}^*_{\mathrm{birth}}\neq\mathsf{a}_{\mathrm{birth}}$ as typed
tuples (different $\mathsf{pk}_0$, sponsor, request, or time window).
\end{definition}

\begin{theorem}[Birth unforgeability]
\label{thm:birthuf}
Under the assumptions of Theorem~\ref{thm:unforge}, and if $H_{\mathsf{msg}}$
is collision-resistant, then for every \PPT{} $\mathcal{A}$ there are
reductions with
\[
  \Adv^{\mathsf{birth}}_{\mathsf{ACE}}(\mathcal{A})
  \le
  \Adv^{\mathsf{ks}}_{\Pi}(\mathcal{B}_1)
  +\Adv^{\mathsf{coll}}_{\mathsf{S}_1}(\mathcal{B}_2)
  +\Adv^{\mathsf{rec}}(\mathcal{B}_3)
  +\Adv^{\mathsf{coll}}_{H_{\mathsf{msg}}}(\mathcal{B}_4).
\]
\end{theorem}
\begin{proof}
\emph{Case $x^*=x$.} \textsc{CreateNativeAccount} recomputes $\msg$
by~\eqref{eq:msg} from $\mathsf{a}^*_{\mathrm{birth}}$ and requires it to
equal the $\msg$ in $x$. Since $\mathsf{a}^*_{\mathrm{birth}}\neq
\mathsf{a}_{\mathrm{birth}}$, this is a collision in $H_{\mathsf{msg}}$, which
$\mathcal{B}_4$ outputs.

\emph{Case $x^*\neq x$.} The proof $\pi^*$ verifies for a statement with
identity $\ID$. Apply the three hops of the proof of
Theorem~\ref{thm:unforge}: extraction of $\REV^*$ (cost $\mathcal{B}_1$);
$\REV^*\neq\REV$ yields a collision in $\mathsf{S}_1$ ($\mathcal{B}_2$);
$\REV^*=\REV$ yields a recovery of $\REV$ from the public view, which here
contains the pending tuple in the role of $\mathcal{H}$ ($\mathcal{B}_3$).
\end{proof}

\subsection{Substitution}
\label{sec:substitution}

\begin{definition}[Substitution]
\label{def:substitution}
Setup as in Definition~\ref{def:unforge}. An honest tuple
$(x,\pi,\mathsf{a})$ is generated and given to $\mathcal{A}$. $\mathcal{A}$
wins if \textsc{VerifyAndExecute} accepts $(x',\pi',\mathsf{a}')$ and either
(i) $\pi'=\pi$ with $x'\neq x$, or
(ii) $\msg'=\msg$ but the typed encoding of
$(\ID',\dom',\opnonce',\tgt',\mathsf{a}')$ differs from the honest encoding.
A new proof at a new digest is an unforgeability question, not a win in
this substitution game.
\end{definition}

\begin{theorem}[Substitution resistance]
\label{thm:substitution}
If $\Pi$ is public-input-binding and $H_{\mathsf{msg}}$ is
collision-resistant, and if the executor recomputes~\eqref{eq:msg}, verifies
the committed executor and authenticated submitter, and dispatches exactly
$\mathsf{calldata}$, then
\[
  \Adv^{\mathsf{sub}}_{\mathsf{ACE}}(\mathcal{A})
  \le
  \Adv^{\mathsf{pib}}_{\Pi}(\mathcal{B}_3)
  +\Adv^{\mathsf{coll}}_{H_{\mathsf{msg}}}(\mathcal{B}_4).
\]
\end{theorem}
\begin{proof}
\emph{Case A: $\pi'=\pi$ and $x'\neq x$.} Then the same proof verifies under
two public inputs. $\mathcal{B}_3$ outputs $(\pi,x,x')$ to the binding
challenger.

\emph{Case B: equal digest, distinct encoded action.} The executor recomputes
\eqref{eq:msg} from $(\ID',\dom',\opnonce',\tgt',\mathsf{a}')$ and requires equality with
$\msg$. Two distinct typed encodings that hash to the same $\msg$ are a
collision; $\mathcal{B}_4$ outputs them.

The honest $\tgt$ and $\opnonce$ are inputs of~\eqref{eq:msg} and are re-checked against
$\mathcal{R}$ and $\mathcal{N}$. The cases are exhaustive for the winning
condition of Definition~\ref{def:substitution}.
\end{proof}

Public-input binding is not knowledge soundness: $x'$ may still be in the
language. A STARK instantiation must absorb $x$ into the Fiat--Shamir
transcript; omitting that step would void Theorem~\ref{thm:substitution}
even if low-degree tests hold.

\subsection{Replay}
\label{sec:replay}

\begin{definition}[Replay]
\label{def:replay}
$\mathcal{A}$ wins if two successful \textsc{CreateNativeAccount} or
\textsc{VerifyAndExecute} transitions share the same $\rp$, or if two
successful \textsc{VerifyAndExecute} transitions share the same sequence
position $(\dom,\ID,\tgt,\opnonce)$.
\end{definition}

\begin{theorem}[Replay resistance]
\label{thm:replay}
Assume the writes in Algorithms~\ref{alg:birth} and~\ref{alg:protocol} are
serializable and atomic with execution.
\begin{enumerate}[nosep,leftmargin=1.5em]
  \item Against reuse of a consumed $\rp$ or a consumed
  $(\dom,\ID,\tgt,\opnonce)$, $\Adv^{\mathsf{replay}}=0$ in the state-machine
  game (excluding crashes that roll back only part of the transaction).
  \item Against minting a \emph{fresh} $\rp$ for the currently expected nonce
  without an opening of the registered $\ID$, the advantage is at most
  $\Adv^{\mathsf{uf}}_{\mathsf{ACE}}$ from Theorem~\ref{thm:unforge}.
\end{enumerate}
\end{theorem}
\begin{proof}
(1)~The first accepting transition inserts $\rp$ into $\mathcal{U}$ and sets
the stored nonce to $\opnonce+1$. Any later transition sees one of those
writes and fails either $\rp\notin\mathcal{U}$ or
$\opnonce=\mathcal{N}[\dom,\ID,\tgt]$.

(2)~(C3) and (C5) bind $\rp$ to $(\REV,\msg,\dom,\opnonce)$, and
$\mathsf{S}_{\mathsf{r}}$ is deterministic, so the holder's own proof for this
slot and $\msg$ carries one $\rp$. Any accepted tuple with a fresh $\rp$ at
the expected nonce is therefore not in $\mathcal{H}$ and is a win in
$\mathsf{Exp}^{\mathsf{uf}}$, whose advantage is bounded by
Theorem~\ref{thm:unforge}.
\end{proof}

Theorem~\ref{thm:replay}(1) takes serializability from the implementation's
transaction layer. Overflow of the $64$-bit counter is a reject, so no
sequence position is reused by modular wrap.

\subsection{Privacy}
\label{sec:privacy}

If $\Pi$ is zero-knowledge, a proof reveals nothing about $w$ beyond $x$. The
values $\ID,\tgt,\dom,\opnonce,\msg,\rp$ remain public, so equal identifiers
are linkable. Account $\ID$ is meant to be public; $\REV$ is not.

The reference Circle STARK fills unused trace rows with fresh field elements,
restricts witness constancy to active rows, uses off-domain openings, and
mixes public inputs into the Fiat--Shamir channel~\cite{stwo}. These are
witness-masking measures in the prototype; this paper does not prove that the
concrete Stwo profile is zero-knowledge. The privacy statement above applies
only to a backend for which zero knowledge has been established.

\section{Instantiation}
\label{sec:implementation}

\subsection{Encodings}
\label{sec:encodings}
The STARK encoding of $x$ is Section~\ref{sec:circuit:enc}. The exact SHA-256
action digest~\eqref{eq:msg} is computed outside the circuit. The executor
recomputes it before dispatch.

\subsection{Algebraic hash}
\label{sec:sponge}
$\mathsf{S}_1$, $\mathsf{S}_{\mathsf{t}}$ and $\mathsf{S}_{\mathsf{r}}$ are
independent invocations of the instance in Table~\ref{tab:poseidon2}.
$\mathsf{S}_{\mathsf{t}}$ and $\mathsf{S}_{\mathsf{r}}$ absorb the distinct
tuples in (C2) and (C3), respectively; they cannot be replaced by one shared
sponge state without changing the relation.

\subsection{Backend}
\label{sec:prototype}
The reference backend uses Stwo 2.2.0~\cite{stwo} over Mersenne-31. Setup is
transparent. Its argument soundness combines hash binding and a low-degree
test; end-to-end soundness additionally requires the AIR to enforce every
constraint of $\mathcal{L}_{\mathsf{ACE}}$. The verifier pins the Merkle root
of the preprocessed columns and mixes $x$ into the Fiat--Shamir transcript.
The FRI configuration is log-blowup $3$, $70$ queries, fold step $1$, last-layer
degree bound $2^1$, and $16$ proof-of-work bits. Stwo's security estimate
$70\cdot 3+16=226$ bits rests on the conjectured proximity-gap behaviour of
FRI. The challenge field, the degree-four extension of Mersenne-31, has about
$2^{124}$ elements, which caps the overall soundness at roughly 124 bits under
that analysis. No provable level is claimed. Proof bytes use an application envelope with
magic \texttt{ZKACEP01}, a SHA-256 profile identifier, an exact little-endian
payload length, and bounded JSON for the pinned Stwo proof type. The verifier
rejects envelopes above 1\,MiB, malformed nested structures, inconsistent FRI
metadata, and non-canonical Mersenne-31 limbs before field arithmetic. The
proof sizes reported below are the lengths of this JSON envelope; a canonical
binary encoding of the same proof would be smaller and is not measured here.
The reference crate is available at~\cite{zkace-crate}.

This is a research profile. Stwo 2.2.0 has no hiding PCS API, and the current
mandatory dependency graph includes a timing advisory in
\texttt{crypto-bigint 0.5.5} plus unmaintained build dependencies. The selected
Blake2s path does not call Stwo's Poseidon252 implementation, but production
release remains blocked pending a reviewed dependency set and independent
review of the composed hash, AIR, Fiat--Shamir, and PCS profile. The
auxiliary-input knowledge-soundness assumption of
Theorem~\ref{thm:unforge} for the concrete Stwo instantiation awaits that
review.

The crate implements the proof relation, bounded verifier, authorization
handler, and an in-memory reference for the atomic transition. It does not
implement the full ledger birth procedure, sponsor policy, durable storage, or
installed-key governance in Algorithm~\ref{alg:birth}. The in-memory reference
retains only the latest active replay value for each
$(\dom,\ID,\tgt)$ binding: exact-next-nonce enforcement makes every older value
non-executable. A production implementation must preserve the same invariant
with durable serializable state or retain the full set $\mathcal{U}$.

The same relation can be proved by a pairing SNARK for an existing pairing
verifier. That profile needs a circuit-specific reference string, is not
post-quantum, and is not measured in this paper. Changing the sponge field
requires re-registration of $\ID$.

\subsection{Local evaluation}
\label{sec:not-reported}
Table~\ref{tab:pq-measure} reports a local release snapshot of the reference
\texttt{zk-ace} crate on an Apple M3 Pro with 36\,GiB RAM, macOS on
\texttt{aarch64}, Rust nightly 1.90.0, Stwo 2.2.0, and the locked dependency
set. Criterion~\cite{criterion} collected 100 samples after a 3-second warm-up. Proving includes
trace construction and serialization; verification includes bounded envelope
decoding and structural validation. These figures are for this working tree
and machine. Recursive aggregation is not measured.

\begin{table}[t]
\centering
\caption{Poseidon2-M31 Circle STARK implementation snapshot.}\label{tab:pq-measure}
\begin{tabular}{@{}lr@{}}
\toprule
Metric & Circle STARK \\
\midrule
Proof (JSON envelope) & $347{,}094$ bytes ($\approx339$\,KiB) \\
Public input & $148$ bytes \\
Compute public input & $13.607$--$13.638$\,$\mu$s \\
Prove & $14.660$--$15.285$\,ms \\
Verify & $2.4445$--$2.4771$\,ms \\
Active rows & $11\times 23=253$ \\
Trusted setup & none \\
\bottomrule
\end{tabular}
\end{table}

\subsection{Ledger placement}
\label{sec:deployment}
Native identity is $\ID=\mathsf{S}_1(\mathsf{Enc}(\REV,\mathsf{s}_{\mathrm{id}},\dom))$, relation (C1).
Account birth is Algorithm~\ref{alg:birth}: the proof shows control of
$\REV$, the sponsor pays, and the executor installs $\mathsf{pk}_0$. After
birth, ordinary transfers verify a key in $\mathcal{K}[\dom,\ID]$ (Ed25519
or ML-DSA-44). Those signatures are microseconds and 64 bytes or about
2.4\,KB. They are the hot path; \ZKACE runs at birth. A transfer to $\ID\notin\mathcal{C}$ rejects.

Table~\ref{tab:artifacts} is the artifact statement. A measured 339\,KiB STARK
is paid once at birth. Repeating it on every transfer would add an object more
than two orders of magnitude larger than ML-DSA-44. Recursion that folds many birth or
aggregation proofs into one on-chain verifier is a deployment target. A tiled batch STARK
remains linear in $N$, and its verification cost is usually greater than
that of $N$ single proofs.

\begin{table}[t]
\centering
\caption{Authorization artifacts on a high-throughput ledger.}\label{tab:artifacts}
\small
\begin{tabular}{@{}lll@{}}
\toprule
Path & Object & When \\
\midrule
Account birth & Circle STARK, $\approx339$\,KiB & once \\
Transfer (Ed25519) & $64$\,B signature & every transfer \\
Transfer (ML-DSA-44) & ${\approx}\,2.4$\,KB signature & every transfer \\
Per-transfer \ZKACE & same STARK as birth & not on the hot path \\
\bottomrule
\end{tabular}
\end{table}

The application that is not the birth path still recomputes the action
digest, checks deadline and dispatch policy, and runs
\textsc{VerifyAndExecute}. Overflow of the final u64 nonce rejects.
The proof crate supplies none of the persistent writes. Changing the
sponge field changes $\ID$ and requires a new birth.

Because $\ID$ is a hash of $\REV$ and not a verification key, replacing an
installed Ed25519 key with an ML-DSA key is an installed-key operation on
$\mathcal{K}[\dom,\ID]$ and leaves $\ID$ unchanged. That rotation, like any
ACE-GF rotation, governs only future use of the installed-key path; it does
not undo an exposure of $\REV$, whose remedy is a new $\REV$ and a new $\ID$.

\subsection{Account-abstraction integration}
\label{sec:deployment:aa}
On a chain with account abstraction~\cite{erc4337}, the application path above
can run as the validation logic of a smart account. The module then performs
the whole of \textsc{VerifyAndExecute}: it recomputes~\eqref{eq:msg} with
$\mathsf{executorId}$ fixed to the account and $\mathsf{submitterId}$ fixed to
the principal that the deployment treats as authenticated (for example the
entry-point contract), checks the time window against the block clock, looks
up $\mathcal{R}$, and requires and writes $\mathcal{N}$ and $\mathcal{U}$ in
the same transaction as the call. A validation step that verifies the proof
while the nonce and replay value are written in a later transaction is not
this protocol. A bundler or relayer only transports $(x,\pi,\mathsf{a})$ and
never holds $\REV$; substitution of the tuple in transit is covered by
Theorem~\ref{thm:substitution}. On an EVM chain the 339\,KiB STARK of
Table~\ref{tab:pq-measure} is not a practical on-chain object, so such a module
would use the pairing-SNARK profile for an existing precompile, with the
reference-string and post-quantum caveats stated in Section~\ref{sec:prototype}.
That profile is not measured here.

\section{Limitations and Future Work}
\label{sec:limitations}

\begin{itemize}[nosep,leftmargin=1.5em]
  \item \textbf{No concrete security level.} The Mersenne-31 profile is not
  assigned a concrete classical or post-quantum level. The 248-bit capacity
  and the roughly 124-bit cap from the challenge field are ceilings under
  stated analyses, not proofs for the composed hash, AIR, Fiat--Shamir
  transform, and PCS (Sections~\ref{sec:security} and~\ref{sec:prototype}).
  \item \textbf{Zero knowledge of the prototype.} Stwo 2.2.0 has no hiding PCS
  API, and this paper does not prove that the concrete profile is
  zero-knowledge. The privacy claim and the zero-knowledge route to
  $\Adv^{\mathsf{rec}}$ apply only to a backend for which zero knowledge is
  established; for the prototype, recovery hardness is an assumption on the
  proof system.
  \item \textbf{No adaptive proving oracle.} The unforgeability experiment
  gives the adversary an honestly generated history but no adaptive
  chosen-action oracle, so it is narrower than EUF-CMA. A proof in that
  stronger model would need simulation-extractability of the backend.
  \item \textbf{Root exposure per proof.} The witness is $\REV$ itself, so the
  root must be in prover memory whenever a birth or \ZKACE action is proved.
  Placing \ZKACE at birth and installed keys on the hot path limits how often
  this happens; a witness that is a per-domain derived secret linked to $\ID$
  is future work.
  \item \textbf{Linkability.} $\ID$ is a one-way image under a public
  constant, and all public inputs are linkable. Anonymous membership is not a
  goal of this protocol.
  \item \textbf{Implementation scope.} The crate does not implement the full
  ledger birth procedure, sponsor policy, durable storage, or installed-key
  governance, and its dependency set awaits review
  (Section~\ref{sec:prototype}).
  \item \textbf{Measurement scope.} The numbers of Table~\ref{tab:pq-measure}
  come from one machine and one working tree, single-threaded. The proof size
  is the JSON envelope; a canonical binary encoding, recursive aggregation,
  cross-platform runs, and the pairing-SNARK profile are not measured.
  \item \textbf{Out-of-scope operations.} Sponsor admission policy, account
  recovery after loss of $\REV$, and governance that later adds or revokes
  installed keys are outside the three games.
\end{itemize}

\section{Conclusion}
\label{sec:conclusion}

\ZKACE is authorization from a registered commitment opening, bound to a
typed action, a public nonce, and a replay value. Birth is
\textsc{CreateNativeAccount}: an atomic write that creates the account,
installs $\mathsf{pk}_0$, and consumes $\rp$. Later calls use
\textsc{VerifyAndExecute} or an installed key. Unforgeability of birth and of
execution with public history reduces to auxiliary-input knowledge soundness,
collision resistance, and recovery hardness (which relies on the proofs hiding
the witness). Substitution is public-input binding plus
collision resistance of the action hash. Replay is a serializable write, with
fresh-token minting reduced to unforgeability.

The circuit is a published Poseidon2-M31 sponge over a short witness. The
measured reference serialization is about 339\,KiB. On a
high-throughput ledger those proofs are for account birth and for optional
aggregation. Transfers keep an installed authentication key.

\section*{Acknowledgments}

The author thanks Alberto La Rocca for discussion on sponge-based circuit
minimization~\cite{elias-personal-2026}, in particular the observation that a
sponge needs no HMAC-style outer hash, which informed the second-lane squeeze
of $\mathsf{S}_{\mathsf{t}}$ and $\mathsf{S}_{\mathsf{r}}$.


\end{document}